\documentclass[runningheads]{llncs}
\usepackage{graphicx}
\usepackage{graphicx}
\usepackage{orcidlink}
\usepackage{amsmath}
\usepackage{bm}
\usepackage{cleveref}
\usepackage{mathpartir}
\usepackage[strict]{changepage}
\usepackage{amsfonts}
\usepackage[inline]{enumitem}
\usepackage{tikz}
\usetikzlibrary{trees}
\usepackage{url}

\usepackage{ cmll }
\usepackage[clock]{ifsym}
\usepackage{empheq}
\usepackage[inline]{enumitem}
\usepackage{centernot}
\usepackage{dashbox}

\usepackage{mathpartir}
\newcounter{rule}
\crefformat{rule}{#2[#1]#3}
\Crefformat{rule}{#2[#1]#3}
\renewcommand{\infer}[5][]{%
    \def\therule{#2}%
    \refstepcounter{rule}%
    \label{rule:#3}%
    \inferrule*[lab={#2}, right={#1}, vcenter]{#4}{#5}%
}

\usepackage{tikz}
\usepackage{pifont}

\newtheorem{innercustomthm}{Theorem}
\newenvironment{customthm}[1]
  {\renewcommand\theinnercustomthm{#1}\innercustomthm}
  {\endinnercustomthm}

\newcommand*{\eg}{\textit{e.g.}}
\newcommand*{\Eg}{\textit{E.g.}}
\newcommand*{\ie}{\textit{i.e.,}}

\newcommand*{\cf}{\textit{cf.}}

\newcommand*{\etal}{\textit{et~al.}}
\newcommand*{\vv}{\textit{vice~versa}}
\newcommand*{\wrt}{w.r.t.}

\definecolor{colabred}{RGB}{242, 60, 60}
\definecolor{type}{RGB}{0, 51, 102}

\newcommand{\type}[1]{{\color{type} #1}}
\let\oldGamma\Gamma
\renewcommand{\Gamma}{{\color{type} \oldGamma}}
\newcommand{\Gammap}{{\color{type} \oldGamma^\prime}}

\let\oldDelta\Delta
\renewcommand{\Delta}{{\color{type} \oldDelta}}
\newcommand{\Deltab}[1]{{\color{type} \oldDelta_{#1}}}
\newcommand{\Deltap}{{\color{type} \oldDelta^\prime}}
\let\oldBeta\Theta
\renewcommand{\beta}{{\color{type} \oldBeta}}
\newcommand{\betap}{{\color{type} \oldBeta^\prime}}
\let\oldTheta\Theta
\renewcommand{\Theta}{{\color{type} \oldTheta}}
\newcommand{\Thetap}{{\color{type} \oldTheta^\prime}}
\newcommand{\Thetab}[1]{{\color{type} \oldTheta_{#1}}}

\newcommand{\magpi}{MAG\ensuremath{\pi}}

\newcommand{\tab}{\ \ \ \ }
\newcommand{\then}{\, .\,}

\newcommand{\net}{\mathcal{N}}
\newcommand{\buf}{\mathcal{B}}
\newcommand{\M}{\mathcal{M}}
\newcommand{\R}{\mathcal{R}}
\newcommand{\F}{\mathcal{F}}

\newcommand{\tl}{\triangleleft}
\newcommand{\cons}{\cdot}
\let\oldAnd\&
\renewcommand{\&}{\mathbin{\oldAnd}}
\renewcommand{\|}{\; ||\;}
\newcommand{\pPar}{\; |\;}
\newcommand{\tPar}{\color{type}\; |\;}
\renewcommand{\mid}{\ \bigm\vert\ }
\renewcommand{\o}{\oplus}
\newcommand{\0}{\bm 0} 
\newcommand{\clock}{\ensuremath{\,{\scriptstyle\VarClock}}}

\newcommand{\tEnd}{\textbf{\color{type}\textsf{end}}}

\newcommand{\role}[1]{\textsf{\color{colabred} \textbf{#1}}}
\newcommand{\msgLabel}[1]{\textsf{#1}}
\newcommand{\Msg}[4]{\langle #1 \rightarrow #2 , #3\langle #4 \rangle \rangle}
\newcommand{\ctxs}[3]{#1 \; ; \; #2 \; ; \; #3}

\newcommand{\basic}[1]{\color{type}\texttt{#1}}

\newcommand{\aOut}[3]{\type{#1{\o}#2{:}#3}}
\newcommand{\aCom}[3]{\type{#1{,}#2{:}#3}}
\newcommand{\aTime}[1]{\type{\clock}\role{#1}}
\newcommand{\redS}[1]{\ensuremath{\xrightarrow{#1}}}
\newcommand{\redP}{\rightarrow}
\newcommand{\redN}[1]{\longrightarrow_{#1}}

\newcommand{\dom}{\textsf{dom}}
\newcommand{\roles}{\textsf{roles}}
\newcommand{\safe}{\varphi}
\newcommand{\pEnd}{\textsf{end}}
\newcommand{\tterm}{\textsf{tt}}
\newcommand{\df}{\textsf{df}}
\newcommand{\term}{\textsf{term}}

\newcommand{\p}{\role{p}}
\newcommand{\q}{\role{q}}
\newcommand{\m}{\msgLabel{m}}
\renewcommand{\l}{\msgLabel{l}}
\newcommand{\St}{\type{S}}
\newcommand{\Rt}{\type{R}}
\newcommand{\ping}{\msgLabel{ping}}
\newcommand{\pong}{\msgLabel{pong}}
\newcommand{\ok}{\msgLabel{ok}}
\newcommand{\ko}{\msgLabel{ko}}
\newcommand{\rr}{\role{r}}
\newcommand{\cc}{\role{c}}
\newcommand{\s}{\role{s}}

\newcommand{\xmark}{{\color{colabred} \ding{55}}}%

\begin{document}
\title{\magpi{!}: The Role of Replication in Typing Failure-Prone Communication}
\titlerunning{\magpi{!}}
%
\author{Matthew Alan Le Brun \and
Ornela Dardha}
\authorrunning{M. A. Le Brun and O. Dardha}
\institute{University of Glasgow, Glasgow, UK\\
\email{\{matthewalan.lebrun,ornela.dardha\}@glasgow.ac.uk}}
\maketitle              
\begin{abstract}

\magpi\ is a Multiparty, Asynchronous and Generalised $\pi$-calculus that introduces timeouts into session types as a means of reasoning about failure-prone communication. 
Its type system guarantees that all possible message-loss is handled by timeout branches. 
In this work, we argue that the previous is unnecessarily strict. 
We present \magpi{!}, an extension serving as the first introduction of replication into Multiparty Session Types (MPST). 
Replication is a standard $\pi$-calculus construct used to model infinitely available servers. 
We lift this construct to type-level, and show that it simplifies specification of distributed client-server interactions. 
We prove properties relevant to generalised MPST: subject reduction, session fidelity and process property verification.

\end{abstract}


\section{The Tale of the MAG(pie/$\pi$)}

The magpie is a bird with deep ties to British folklore. 
The first known mention of their counting for fortune telling dates back to 1780, where John Brand writes what is thought to be one of the original versions of the magpie rhyme~\cite{brand1841observations}:
\begin{center}
    ``\textit{One for sorrow, Two for mirth, Three for a funeral, And four for a birth.}''
\end{center}
We can imagine that the natural reaction of a person who spots a solitary magpie is to scan the surrounding area for its companion.
Alas, if no one is immediately visible, the person desperately waits---hoping a second magpie comes their way.
But how long should one wait? 
The reality is that it is \emph{impossible} to know the difference between \emph{no magpie} and a magpie that has \emph{not yet arrived}.
To computer scientists, this is a well known \emph{impossibility result}~\cite{DBLP:conf/sosp/AkkoyunluEH75}.
In the study of \emph{distributed systems} and \emph{fault tolerance}, mechanisms must be employed to approximate the impossibility result of determining whether a message has been \emph{lost} or \emph{delayed}---\eg\ by using a \emph{timeout}.
Hence, the computer scientist who spots a lonely magpie knows to only wait some fixed amount of time before \emph{assuming} that no other magpie is coming and accepting their sorrowful faith.
This philosophy is the core principle of the process calculus \magpi~\cite{DBLP:conf/esop/BrunD23}, a language designed to model communication failures (via \emph{message loss}) with a generic type system aiming to provide configurable runtime guarantees.

\magpi\ is a Multiparty, Asynchronous and Generalised $\pi$-calculus, modelling distributed communication over $n$-participant \emph{sessions}. 
Its key features include non-deterministic failure injection into the runtime of a program, asynchronous communication via \emph{bag} buffers (allowing for total message reordering), and a generic type system capable of providing guarantees of runtime properties via \emph{session types}.
Session types~\cite{DBLP:conf/esop/HondaVK98,DBLP:journals/iandc/Vasconcelos12,DBLP:journals/iandc/DardhaGS17} are \emph{behavioural type systems} allowing for formal specification of communication protocols---their main benefit being that they provide correctness guarantees on both protocol design and implementation.
\emph{Multiparty session types} (MPST)~\cite{DBLP:conf/popl/HondaYC08,DBLP:conf/concur/BettiniCDLDY08,DBLP:journals/pacmpl/ScalasY19} are a branch of session type theory that aims to support protocols involving \emph{any} number of participants with interleaving communication.
\magpi\ builds upon a generalised form of MPST~\cite{DBLP:journals/pacmpl/ScalasY19,DBLP:conf/concur/BarwellSY022}, where protocols are defined by a collection of \emph{local types}---the communication patterns of individual participants' perspectives---which should be exhaustively checked (\eg\ via model checking) to determine any properties they observe.
Novelties of \magpi\ stem from how it embraces the impossibility result of distinguishing between dropped or delayed messages; its language and type system use non-deterministic timeouts to model the \emph{assumption} of failures.
The type system guarantees that \emph{all} failure-prone communication is handled by a timeout branch.
In this work, we argue that the previous approach can, in some scenarios, be unnecessarily strict---resulting in needlessly more complex protocols.
Some configurations may wish to leave the handling of failures up to senders, as opposed to recipients; these usually take the form of \emph{client-server} interactions where servers are designed to remain infinitely available.
\Eg\ if a request to a web-server were to drop, it is the client's responsibility to re-issue that request.
We present an extension to \magpi\ that better models infinitely available servers and simplifies failure-handling for client-server interactions.

In the $\pi$-calculus~\cite{DBLP:books/daglib/0004377}, a standard construct often used for representing infinite behaviour is that of \emph{replication}.
A replicated process is one which can be informally described as \emph{infinitely available}.
Naturally, the use of replicated processes lends itself well to the modelling of client-server interactions.
We demonstrate how the use of replication in \magpi\ can, not only better model infinitely available servers, but also simplify their protocols by relaxing the requirement of failure-handling branches from \emph{every} receive to only \emph{linear} receives.

\begin{example}[Type-level replication]\label{ex:motivation}
We evolve the motivating example presented in~\cite[Ex. 1]{DBLP:conf/esop/BrunD23}, the \emph{ping} protocol.
Consider three participants: client \cc,  server \s, and result channel \rr.
Communication between \cc\ and \rr\ is reliable; whereas with \s\ is \emph{unreliable}.
The session types for a three-attempt ping in \magpi{!} are:
{\small \[\begin{array}{l l}
        \St_\rr &= \type{\&\{\cc : \ok \then \tEnd,\ \cc : \ko \then \tEnd\}} \\
        \St_\cc &= \type{\o\, \s : \ping \then\!\! \&\!\! \left\{\begin{array}{l}
            \s : \pong \then \o \rr : \ok \then \tEnd, \\
            \clock \then\! \o \s : \ping \then\!\! \&\!\! \left\{\begin{array}{l}
                {\s : \pong \then \o \rr : \ok \then \tEnd}, \\
                \clock \then\! \o \s : \ping \then\!\! \&\!\! \left\{\begin{array}{l}
                    {\s : \pong \then\!\! \o \rr : \ok \then \tEnd}, \\
                    {\clock \then\! \o \rr : \ko \then \tEnd}
                \end{array}\right.
            \end{array}\right.
        \end{array}\right.}\\
        \St_\s &= \type{{!}\cc : \ping \then \o \cc : \pong \then \tEnd}
    \end{array}
\]}

\noindent
Client \cc\ sends a message with label \ping\ to server \s\ (\type{$\o\, \s : \ping$}) and waits for a \pong\ response (\type{$\& \s : \pong$}).
If successful, an \ok\ message is sent to results role \rr\ and the session is terminated for the client (\tEnd).
Since communication with the server is \emph{unreliable}, receipt of the \pong\ message is not guaranteed, and must be handled by a \emph{timeout} branch \type{$\clock$}.
The client attempts to reach the server 3 times---if all attempts fail, it sends a \ko\ message to \rr.
The result role \rr\ waits for either of the reliable responses from \cc, thus no timeout is defined.
Server \s\ is defined as the replicated receive \type{${!}\cc : \ping \then \o \cc : \pong \then \tEnd$}, denoting its constant availability to receive a \ping\ request and send a \pong\ response.
We highlight the absence of a failure-handling timeout branch in \type{$\St_\s$}; the server does not need to change its behaviour if a client request fails.
Furthermore, if the \pong\ reply fails, the server remains available to handle any number of retries from the client.
Thus, the use of replication has offloaded the handling of failures entirely onto the client-side, has made the protocol more modular (since the type for \s\ is now agnostic of a client's retry limit), and is simpler \wrt\ to the original specification in \magpi\ (a full comparison is made in \cref{app:simplify}).

\end{example}

\subsubsection{Contributions.} 
Concretely, our contributions are as follows:
\begin{enumerate}
    \item \textbf{\magpi{!} Language}:
    We present \magpi{!} (\cref{sec:lang}), an extension of \magpi\ that does away with recursion in favour of replication as a better means of modelling client-server interactions.

    \item \textbf{\magpi{!} Types}:
    We lift replication to type-level in \cref{sec:types}.
    To the best of our knowledge, this work serves as the \emph{first introduction} of replication into MPST.
    We improve upon the theory of \magpi\ and show how three type contexts (\emph{unrestricted}, \emph{linear} and \emph{affine}) can be used to type---and \emph{simplify}---failure-prone communication in client-server interactions.

    \item \textbf{\magpi{!} Meta-Theory}:
    \Cref{sec:meta} expounds upon the meta-theory of our type system.
    We prove \emph{subjection reduction} and \emph{session fidelity}, and demonstrate how they can be used for \emph{property verification}.
    \magpi{!} provides a \emph{failure handling guarantee}, ensuring all failure-prone communication is handled by a timeout branch---a responsibility which servers offload to clients.
\end{enumerate}
In \cref{sec:conc} we conclude and give an account of related and future work.

\subsubsection{On delegation and language simplification.}
This work builds upon a \emph{subset} of \magpi~\cite{DBLP:conf/esop/BrunD23} as our language only considers communication over a \textit{single session}.
Reasons for this are: 
\begin{enumerate*}[label=(\textit{\roman*})]
    \item to simplify notation for better readability due to limited space; and
    \item to remove session fidelity assumptions.
\end{enumerate*}
On the latter, generalised MPST theory assumes communication over a single session to prove \emph{session fidelity} (a.k.a. protocol compliance)~\cite[Def. 5.3]{DBLP:journals/pacmpl/ScalasY19}.
This is to remove deadlocks that can occur due to incorrect interleaving of multiple sessions.
Effectively, the language subset we consider syntactically abides by the assumptions of session fidelity by assuming all communication happens over a single session and by removing delegation.
We foresee no issues with extending \magpi{!}\ to multiple sessions, although this will only improve the number of \emph{safe} protocols that can be expressed and has \emph{no effect} on verification of other properties.
Lastly, replication in \magpi{!}\ is a \emph{top-level} construct only.
This simplifies our type-system at the cost of sacrificing expressivity of nested replication.
The type system can still express meaningful examples (\eg\ load balancers), and we intend to explore hidden and nested replication in future work.

\section{Bird Songs}\label{sec:lang}

We present \magpi{!}, an extension of \magpi\ that replaces recursion with replicated processes as its preferred means of reasoning about infinite behaviour.
Programs in \magpi{!} represent distributed networks, consisting of concurrent and parallel processes running on machines connected over some \emph{failure-prone} medium.
We discuss how networks of various topologies are defined in \cref{sec:topology}.
\Cref{sec:processes} details the syntax and semantics of processes.

\subsection{Topology}\label{sec:topology}

Distributed protocols typically consist of a number of participants (or \emph{roles}) representing physically separated devices, communicating over a \emph{failure-prone} network.
We model such a setting by associating processes to uniquely identifiable roles, which communicate asynchronously through a \emph{bag buffer} allowing for \emph{total message reordering}.
Roles are related through a notion of \emph{reliability}, modelling physical locations of processes---\ie\ reliable roles are ones that live on the same physical device and thus are not susceptible to communication errors.
A formal account of networks, buffers and reliability is given below.

\subsubsection{Networks.}

A program in \magpi{!} models some distributed network $\net$.
These networks consist of a parallel composition of processes, each representing specific \emph{roles} in the network.
The formal description of a network is given by \cref{def:network}.

\begin{definition}[Networks]\label{def:network}
    A network $\net$ is given by the following grammar:
    \[\net ::= \p \tl \mathcal{P} \mid \net \| \net \mid \buf\]
    where $\buf$ is a message buffer; $\mathcal{P}$ is the process instruction; and \p\ is a role name.
\end{definition}

\noindent
A \emph{process} $\p \;\tl\; \mathcal{P}$ consists of a uniquely identifying role name \p, and process instructions $\mathcal{P}$.
It is key to note that all processes, \ie\ participants, of a network are syntactically defined---thus, \magpi{!} assumes a finite network size where all participants are \emph{statically} known.
The $\|$ constructor denotes \emph{parallel composition} of processes within a network, and $\buf$ is its message buffer.

\subsubsection{Buffers.}

\magpi{!} models asynchrony through a \emph{bag buffer} (semantics discussed in \cref{sec:processes}).
The buffer, \cref{def:buffer}, serves two purposes.
Firstly, it allows for non-blocking (\emph{fire and forget}) sends by acting as an intermediary where messages wait until recipients are ready to consume them.
Second, and important to distributed communication, is that it models messages \emph{in transit} over the network and is thus the point-of-failure in our system.

\begin{definition}[Buffers]\label{def:buffer}
    A message $\M$ is defined as \textnormal{$\M ::= \langle \role{p}\rightarrow\role{q}, \m\langle \tilde{v} \rangle \rangle$}, \ie\ a tuple identifying the source and destination of the message \textnormal{($\role{p}\rightarrow\role{q}$)}, along with a message label and payload contents \textnormal{($\m\langle \tilde{v} \rangle$)}.
    A buffer $\buf$ is a multiset of messages $\M$.
    Concatenating a message $\M$ with a buffer $\buf$, written $\buf \cons \M$ corresponds to the multiset sum of $\buf + \{\M\}$.
\end{definition}

\subsubsection{Reliability.}

A network is initialised with a reliability relation $\R$ (\cref{def:reliability}),
defining roles which may communicate sans failure.
All communication outwith the reliability relation is considered failure-prone; 
this may be used to simulate physical topologies, or to study a protocol at various degrees of reliability.

\begin{definition}[Reliability]\label{def:reliability}
    Given a network $\net$, and set of roles $\rho$ acting in $\net$,
    the reliability relation $\R$ is a subset of (or equal to) $\{\{\p,\q\} : \p,\q \in \rho \,\land\, \p \neq \q\}$.
    We write $\net :: \R$ to denote a network $\net$ governed by reliability relation $\R$.
    We use shorthand $\net :: \F$ to denote a fully reliable network, and $\net :: \emptyset$ to denote a fully unreliable network. 
\end{definition}

\begin{example}[Load Balancer: Network]\label{ex:network}
    Consider a load balancer network with server \role s, workers \role{w$_1$}, \role{w$_2$}, and client \role c.
    Assuming server-worker communication to be reliable, the network may be configured as below:
    \[
    \role{s} \tl \mathcal{P_\role{s}} \|  \role{w$_1$} \tl \mathcal{P_\role{w$_1$}} \| \role{w$_2$} \tl \mathcal{P_\role{w$_2$}}
    \| \role{c} \tl \mathcal{P_\role{c}} \| \buf :: \{\{ \role s, \role{w$_1$}\}, \{ \role s, \role{w$_2$}\}\}
    \]
\end{example}

\subsection{Processes}\label{sec:processes}

\begin{definition}[Process syntax]
    The syntax for defining process instructions $\mathcal{P}$ is given by the following grammar:
    {\normalfont \begin{align*}
        \mathcal{P} &::= {!}_{i \in I} \role{p$_i$} : \m_i (\tilde{x_i}) \then P_i \mid \dbox{$\mathcal{P} \pPar \mathcal{P}$} \mid P\\
        P    &::= \0 \mid \&_{i \in I} \role{p$_i$} : \m_i (\tilde{x_i}) \then P_i\,[, \clock \then P^\prime] \mid \o\, \p : \m \langle \tilde{c} \rangle \then P\\
        c    &::= x \mid v
        \tab\tab 
        v ::= \text{basic values}
    \end{align*}}

    \noindent
    All branching terms assume $I \neq \emptyset$ and all couples $\role{p$_i$} : \m_i$ to be pairwise distinct. 
    Receiving constructs act as binders on their payloads.
    
\end{definition}

\noindent
A process $\mathcal{P}$ can either be a replicated server or a linear process.
\emph{Replicated receive} ${!}_{i \in I} \role{p$_i$} : \m_i (\tilde{x_i}) \then P_i$ denotes a server constantly available to receive any of a set of messages from roles \role{p$_i$} with labels $\m_i$.
The received payload is bound to $\tilde{x_i}$ before pulling out a copy of $P_i$ to run in parallel with the server.
\emph{Parallel composition} $\pPar$ is a \dbox{runtime} only construct at the process-level.
It is used to denote composition of linear continuations pulled out of a replicated receive.
Linear processes ($P, Q, \dots$) consist of:
\begin{enumerate*}[label=(\textit{\roman*})]
    \item the \emph{empty process} $\0$;
    \item \emph{linear receives} $\&_{i \in I} \role{p$_i$} : \m_i (\tilde{x_i}) \then P_i\,[, \clock \then P^\prime]$, where a role waits for one of a set of messages from some other roles \role{p$_i$} with labels $\m_i$, binding the received payload to $\tilde{x_i}$ before proceeding according to $P_i$; 
    \item an optional \emph{nondeterministic timeout branch} $[, \clock \then P^\prime]$ attached to linear receives to handle possible failure of messages, instructing the process to proceed according to $P^\prime$; and
    \item \emph{linear sends} $\o\, \p : \m \langle \tilde{c} \rangle \then P$ which sends a message towards $\p$ with label $\m$ and payload $\tilde{c}$ before continuing according to $P$.
\end{enumerate*}
A payload $c$ is either a \emph{variable} ($x,y,\dots$) or some assumed basic value (integers, reals, strings, \dots).
We omit conditional branching constructs such as if-then-else and case statements as they are routine and orthogonal to our work (we assume them in examples).

\begin{figure}[t]
    \fbox{Process semantics}
    {\normalfont \begin{mathpar}
        \infer{P-Send}{p-send}
        {}
        {\p \tl \q \o \m \langle \tilde{v} \rangle \then P \| \buf \redN{\R}  \p \tl P \| \buf \cons \Msg{\p}{\q}{\m}{\tilde{v}}}

        \infer{P-Recv}{p-recv}
        {\exists k \in I $ s.t. $\role{r$_k$} = \q  \text{ and } \m_k = \l \text{ and } |\tilde{y}_k| = |\tilde{v}|}
        {\p \tl \&_{i \in I} \role{r$_i$} : \m_i (\tilde{y}_i) \then P_i \,[, \clock \then  P^\prime] \| \buf \cons \Msg{\q}{\p}{\l}{\tilde{v}} \redN{\R} \p \tl P_k \{^{\tilde{v}}/_{\tilde{y}_k}\} \| \buf}
        
        \infer{P-$!$Recv}{p-bang}
        {\mathcal{P} = {!}_{i \in I} \role{r$_i$} : \m_i (\tilde{y}_i) \then P_i \\ \exists k \in I $ s.t. $\role{r$_k$} = \q  \text{ and } \m_k = \l  \text{ and } |\tilde{y}_k| = |\tilde{v}|}
        {\p \tl \mathcal{P} \| \buf \cons \Msg{\q}{\p}{\l}{\tilde{v}} \redN{\R} \p \tl \mathcal{P} \pPar P_k \{^{\tilde{v}}/_{\tilde{y}_k}\} \| \buf}

        \infer{N-$\|$}{n-par}
        {\net \redP \net^\prime}
        {\net \| \net^{\prime\prime} \redP \net^\prime \| \net^{\prime\prime}}

        \infer{P-$\pPar$}{p-par}
        {\net \| \p \tl \mathcal{P} \redP \net^\prime \| \p \tl \mathcal{P}^\prime}
        {\net \| \p \tl \mathcal{P} \pPar \mathcal{P}^{\prime\prime} \redP \net^\prime \| \p \tl \mathcal{P}^\prime \pPar \mathcal{P}^{\prime\prime}}
    \end{mathpar}}

    \fbox{Failure semantics}
    {\normalfont \begin{mathpar}
        \infer{F-Drop}{f-drop}
        { \{\p,\q\} \not\in \R }
        {\buf \cons \Msg{\p}{\q}{\m}{\tilde{v}} \redN{\R} \buf}

        \infer{F-Timeout}{f-timeout}
        {}
        {\q \tl \&_{i \in I} \role{r$_i$} : \m_i (\tilde{y}_i) \then P_i, \clock \then P^\prime \redN{\R} \q \tl P^\prime}
    \end{mathpar}}
\caption{Network semantics.}
\label{fig:semantics}
\end{figure}

\begin{definition}[Network Semantics]\label{def:semantics}
    Reduction on networks is parametric on a reliability relation $\R$. 
    The reduction relation $\redN{\R}$ is inductively defined by the rules listed in \cref{fig:semantics}, up-to congruence (\cref{app:cong}).
\end{definition}

Network dynamics (\cref{fig:semantics}) are divided into \emph{process} and \emph{failure} semantics.
A process sends a message via rule \cref{rule:p-send}, which places the message in the network buffer and advances the sending process to its continuation.
Conversely, processes receive messages (rule \cref{rule:p-recv}) by consuming a message from the buffer, advancing the process to its continuation and substituting bound payloads with the received data.
In a similar manner, servers may consume messages from the buffer using rule \cref{rule:p-bang}; instead of advancing the process, a copy of its continuation is \emph{pulled out} and placed in parallel.
This allows servers to concurrently handle and receive client requests.

Message failure is modelled through rule \cref{rule:f-drop}.
We recall that buffers model messages \emph{in transit}, thus this rule may---\emph{at any time}---drop a message from the buffer if it is unreliable.
It is key to note that failure in these semantics is \emph{nondeterministic}.
A client may consume a message before it is dropped, representing a successful transmission; or the message may be dropped before consumed, representing the failure case.
Reduction of \emph{timeout branches} is also nondeterministic since it is impossible to distinguish between \emph{dropped messages} (\emph{no magpie}) and \emph{delayed messages} (the magpie that has \emph{not yet arrived}).
Therefore, rule \cref{rule:f-timeout} can \emph{at any time} reduce a waiting process to its timeout branch, modelling either the handling of message failure or an incorrect assumption of failure (\ie\ message delay).

\begin{example}[Load Balancer: Processes]\label{ex:processes}
    We present the processes of our load balancer. 
    An output role \role o, which is reliable \wrt\ the client, has been added.
    {\small \begin{align*}
        \role{s} &\tl {!} \role{c} : \msgLabel{req}(x) \then \texttt{case flip() of } \left\{\begin{array}{l}
            \texttt{heads} \rightarrow \o \role{w$_1$} : \msgLabel{req}\langle x\rangle \then \0 \\
            \texttt{tails} \rightarrow \o \role{w$_2$} : \msgLabel{req}\langle x\rangle \then \0
        \end{array} \right.  \\
        \role{w$_1$} &\tl {!} \role{s} : \msgLabel{req}(d) \then \o \role{c} : \msgLabel{ans}\langle f(d)\rangle \then \0\\
        \role{w$_2$} &\tl {!} \role{s} : \msgLabel{req}(d) \then \o \role{c} : \msgLabel{ans}\langle f(d)\rangle \then \0\\
        \role{c} &\tl \o \role{s} : \msgLabel{req}\langle 42 \rangle \then \& \left\{\begin{array}{l}
            \role{w$_1$} : \msgLabel{ans}(y) \then \o \role{o} : \msgLabel{output}\langle y \rangle \then \0 \\
            \role{w$_2$} : \msgLabel{ans}(y) \then \o \role{o} : \msgLabel{output}\langle y \rangle \then \0 \\
            \clock \then \o \role{o} : \msgLabel{err}\langle \text{``Request timed out''} \rangle \then \0
        \end{array} \right. \\
        \role{o} &\tl \& \{\role{c} : \msgLabel{output}(out) \then \0, \role{c} : \msgLabel{err}(msg) \then \0\}
    \end{align*}}
\end{example}

\begin{example}[Interactions with Failure: Processes]\label{ex:proc-failures}
    Now we demonstrate interactions unique to our language which result from the use of timeouts as imperfect failure detectors.
    Consider the following network snippet $\net_f :: \emptyset$:
    \[
        \p \tl \o\, \q : \m\langle 42 \rangle \then P\ \|\ \q \tl \& \{\p : \m(x) \then P^\prime,\ \clock \then P^{\prime\prime}\}\ \|\ \{\Msg{\p}{\q}{\m}{\text{``Life is''}}\}
    \]

    \noindent
    These processes denote communication between two roles (\p\ and \q), where a message labelled \m\ with the string ``Life is'' has already been sent, and a second message \emph{also} labelled \m\ is to be sent with payload $42$.
    There are \emph{four} possible immediate reduction steps for this network: 
    \begin{enumerate*}[label=(\textit{\roman*})]
        \item role \q\ consumes the message in the buffer via \cref{rule:p-recv} (the intended behaviour);
        \item role \p\ places message $\Msg{\p}{\q}{\m}{42}$ in the buffer via \cref{rule:p-send}, this may possibly result in message reordering due to the bag buffer semantics;\label{item:f-2}
        \item message $\Msg{\p}{\q}{\m}{\text{``Life is''}}$ is dropped from the buffer via \cref{rule:f-drop}, then \q\ may either correctly assume failure through a timeout, or if the sender is quick enough the message $\Msg{\p}{\q}{\m}{42}$ could still be received in its place; and \label{item:f-3}
        \item role \q\ can incorrectly assume a failure and timeout via \cref{rule:f-timeout} even though message $\Msg{\p}{\q}{\m}{\text{``Life is''}}$ is in the buffer.\label{item:f-4}
    \end{enumerate*}
    It is not difficult to see how \cref{item:f-2,item:f-3,item:f-4} may lead to errors.
    Our types and meta-theory mitigate the occurrence of these possibly unsafe networks by enforcing a safe design of protocols.
\end{example}

\section{Harmonisation}\label{sec:types}

We now present the \emph{multiparty}, \emph{asynchronous}, and \emph{generalised} type system for \magpi{!}.
To the best of our knowledge, this is the first work to introduce \emph{replication} and \emph{parallel composition} for local types in MPST.
We show how these constructs lend themselves well to typing distributed client-server interactions.

\subsection{Types}

The syntax for \magpi{!} types are given in \cref{def:types}.
Our type system does away with tail-recursive binders (as is standard in MPST), instead opting for a \emph{replicated receive} type.
The syntax distinguishes between different classes of types.
Namely, we present \emph{replicated-}, \emph{session-}, \emph{message-} and \emph{basic-types}---each of which are used differently by the type contexts (\cref{def:contexts}).

\begin{definition}[Types]\label{def:types}
    The syntax for \magpi{!} types is given by:
    {\normalfont \begin{align*}
        \type{R} &::= \type{{!}_{i \in I} \role{p$_i$} : \m_i (\tilde{B_i}) \then S_i} \\
        \type{S} &::= \type{\o_{i \in I} \role{p$_i$} : \m_i (\tilde{B_i}) \then S_i} \mid \type{\&_{i \in I} \role{p$_i$} : \m_i (\tilde{B_i}) \then S_i\, [,\clock \then S^\prime]} \mid \dbox{\type{$S \tPar S$}} \mid \tEnd\\
        \type{M} &::= \type{(\p\rightarrow\q, \m(\tilde{B}))} \\
        \type{B} &::= \basic{Int, Real, String,\dots} \text{ (basic types)}
    \end{align*}}

    \noindent
    Branching constructs assume $I \neq \emptyset$ and couples $\role{p$_i$} : \m_i$ to be pairwise distinct.
    Replicated types $\type{R}$ assume a pool of labels distinct from their continuations.
\end{definition}

A \emph{replicated type} \type{$R$} defines the protocol of a server.
Type $\type{{!}_{i \in I}\role{p$_i$} : \m_i (\tilde{B_i}) \then S_i}$ denotes the receipt of requests labelled $\m_i$ from \role{p$_i$} carrying payload types \type{$\tilde{B_i}$} having continuation types \type{$S_i$}.
Replicated types never appear guarded and always have linear continuations.

\emph{Session types} \type{$S$} describe the protocol of a \emph{linear} process.
The \emph{selection} and \emph{branching} types (\type{$\o$} and \type{$\&$}) detail possible sends and receives, indicating direction and content of payloads.
Branching types may optionally include a failure-handling \emph{timeout branch} $\type{\clock \then S}$, where \type{$S$} details the protocol to employ upon assuming a failure.
As in processes, types also have a notion of \dbox{runtime} only \emph{parallel composition}, identifying the protocols of continuations pulled out of a replicated receive.
The \tEnd\ type denotes termination of a party's protocol.

\emph{Message types} \type{$M$} are used to type messages in a buffer. 
They record the direction of communication, as well as the chosen branching label and types of its payload.
Lastly, \type{$B$} represents a range of assumed \emph{basic types}.

\begin{figure}[t]
    \fbox{Context update}
    \begin{mathpar}
        \infer{}{+-0}
        { }
        {\Delta + \emptyset = \Delta}

        \infer[\textnormal{if $\p \not \in \dom(\Deltab{1})$}]{}{+-1}
        {\Deltab{1} + \Deltab{2} = \Deltab{3}}
        {\Deltab{1} + \Deltab{2},\p:\St  = \Deltab{3},\p:\St}

        \infer[]{}{+-2}
        {\Deltab{1} + \Deltab{2} = \Deltab{3}}
        {\Deltab{1},\p:\type{S_1} + \Deltab{2},\p:\type{S_2}  = \Deltab{3},\p:\type{S_1 \tPar S_2}}
    \end{mathpar}

    \fbox{Context splitting}
    \begin{mathpar}
        \infer{}{split-0}
        { }
        {\emptyset = \emptyset \cons \emptyset}

        \infer{}{split-l}
        {\Delta = \Deltab{1} \cons \Deltab{2}}
        {\Deltab,\p:\St = \Deltab{1},\p:\St \cons \Deltab{2}}

        \infer{}{split-r}
        {\Delta = \Deltab{1} \cons \Deltab{2}}
        {\Deltab,\p:\St = \Deltab{1} \cons \Deltab{2},\p:\St}

        \infer{}{split-par}
        {\Delta = \Deltab{1} \cons \Deltab{2}}
        {\Delta,\p:\type{S_1\tPar S_2} = \Deltab{1},\p:\type{S_1} \cons \Deltab{2},\p:\type{S_2}}
    \end{mathpar}
    
    \caption{Context addition and splitting.}
    \label{fig:context-operations}
\end{figure}

\begin{definition}[Contexts]\label{def:contexts}
    Context $\Gamma$ is \textbf{unrestricted} and maps variables to basic types and roles to replicated types.
    Context $\Delta$ is \textbf{linear} and maps roles to session types.
    Context $\Theta$ is \textbf{affine} and holds a multiset of message types $\type{M}$.
    \[ 
        \Gamma ::= \emptyset \mid \p : \type{R}, \Gamma \mid x : \type{B}, \Gamma \tab\tab
        \Delta ::= \emptyset \mid \p : \type{S}, \Delta \tab\tab
        \Theta ::= \{\type{M_1},\dots,\type{M_n}\}
    \]

    \noindent
    \textbf{Updating} and \textbf{splitting} operations are defined for $\Delta$ by the rules in \cref{fig:context-operations}.
    Context \textbf{composition} $\Gamma,\Gammap$ (resp. $\Delta,\Deltap$) is defined iff $\textnormal{\dom}(\Gamma) \cap \textnormal{\dom}(\Gammap) = \emptyset$ (resp. $\textnormal{\dom}(\Delta) \cap \textnormal{\dom}(\Deltap) = \emptyset$).
\end{definition}

\Cref{fig:context-operations} defines two relations on $\Delta$.
\emph{Context addition} joins two contexts by performing a union on their contents (in the case  that there are no conflicts in their domains).
If their domains are not unique, then the types are placed in parallel, indicating a role employing multiple active session types (this is explained in more detail after introducing context reduction \cf\ \cref{def:ctx-red}).
Context \emph{splitting} extracts a piece of a larger context.
Notably, types placed in parallel may be split using this operation; in other cases splitting functions similar to context composition.

\begin{figure}[t]
    \begin{mathpar}
        \infer{$\Delta$-\type{$\clock$}}{D-time}
        { }
        {\ctxs{\Gamma}{ \Delta \cons \p : \type{\&_{i \in I} \role{q$_i$} : \m_i (\tilde{B_i}) \then S_i, \clock \then S^\prime}}{\beta} \redS{\aTime{p}} \ctxs{\Gamma}{ \Delta + \p : \type{S^\prime}}{\beta}}
        
        \infer{$\Delta$-\type{$\o$}}{D-send}
        {\type{S} = \type{\o_{i \in I} \role{q$_i$} : \m_i (\tilde{B_i}) \then S_i} \\ k \in I}
        {\ctxs{\Gamma}{ \Delta \cons \p : \type{S}}{\beta} \redS{\aOut{\p}{\role{q$_k$}}{\m_k}} \ctxs{\Gamma}{ \Delta + \p : \type{S_k}}{\beta \cons \type{(\p \rightarrow \role{q$_k$}, \m_k(\tilde{B_k}))}}}

        \infer{$\Delta$-\type{C}}{D-com}
        {\type{S} = \type{\&_{i \in I} \role{q$_i$} : \m_i (\tilde{B_i}) \then S_i [, \clock \then S^\prime]} \\ k \in I}
        {\ctxs{\Gamma}{ \Delta \cons \p : \type{S}}{\beta \cons \type{(\role{q$_k$} \rightarrow \role{p}, \m_k(\tilde{B_k}))}} \redS{\aCom{\role{q$_k$}}{\p}{\m_k}} \ctxs{\Gamma}{ \Delta + \p:\type{S_k}}{\beta} }

        \infer{$\Gamma$-\type{!C}}{G-bang}
        {\type{R} = \type{{!}_{i \in I} \role{q$_i$} : \m_i (\tilde{B_i}) \then S_i} \\ k \in I}
        {\ctxs{\Gamma, \p : \type{R} }{ \Delta }{\beta \cons \type{(\role{q$_k$} \rightarrow \role{p}, \m_k(\tilde{B_k}))}} \redS{\aCom{\role{q$_k$}}{\p}{\m_k}} \ctxs{\Gamma, \p : \type{R}}{\Delta + \p : \type{S_k}}{\beta} }
    \end{mathpar}
    \caption{Type LTS}
    \label{fig:type-lts}
\end{figure}
\begin{definition}[Context Reduction]\label{def:ctx-red}
    An action $\alpha$ is given by
    {\normalfont\[\alpha ::= \aOut{\p}{\q}{\m} \mid \aCom{\p}{\q}{\m} \mid \aTime{p}\]}

    \noindent
    read as (left to right) \textbf{output}, \textbf{communication}, and \textbf{timeout}.
    Context \textbf{transition} $\redS{\alpha}$ is defined by the Labelled Transition System (LTS) in \cref{fig:type-lts}.
    Context \textbf{reduction} \textnormal{$\ctxs{\Gamma}{\Delta}{\Theta} \redP \ctxs{\Gamma}{\Deltap}{\Thetap}$} is defined iff \textnormal{$\ctxs{\Gamma}{\Delta}{\Theta} \redS{\alpha} \ctxs{\Gamma}{\Deltap}{\Thetap}$} for some \textnormal{$\alpha$}.
    We write $\ctxs{\Gamma}{\Delta}{\Theta} \redP$ iff $\exists {\Deltap},{\Thetap} \text{ s.t. } \ctxs{\Gamma}{\Delta}{\Theta} \redP \ctxs{\Gamma}{\Deltap}{\Thetap}$; and $\redP^*$ for its transitive and reflexive closure.
\end{definition}

Context reduction (\cref{def:ctx-red}) models type-level communication by means of the LTS in \cref{fig:type-lts}.
Transition \cref{rule:D-time} allows a role \p\ with a defined timeout to transition to the timeout continuation by firing a $\aTime{p}$ action.
Transition \cref{rule:D-send} is a synchronisation action between a selection type and the type buffer $\Theta$.
Effectively, a role with a send type can transition to its continuation by firing any of the paths indicated in the selection ($\aOut{\p}{\role{q$_k$}}{\m_k}$) and adding the message into the buffer context.
On the receiving end, a role with a branch type can consume a message from the type buffer to model a communication action via transition \cref{rule:D-com}.
Communication with replicated servers is handled seperately by transition \cref{rule:G-bang}.
This rule allows a communication action to be fired when a replicated type in $\Gamma$ can receive a message in the buffer.
This transition has no effect on $\Gamma$ (since it is an unrestricted context) and instead updates the linear context $\Delta$ with the continuation of the replicated receive.
This is why types require runtime parallel composition, and context updating and splitting operations (\cref{fig:context-operations}), as multiple requests may be made to a replicated receive.

\subsection{Typing Rules}

Protocols defined in \magpi{!} types are used in type judgements (\cref{def:typing-rules}) to check whether network implementations conform to their specifications.

\begin{definition}[Typing Judgement]\label{def:typing-rules}
    Type contexts are used in judgements as $\ctxs{\Gamma}{\Delta}{\Theta} \vdash \net$, inductively defined by the rules in \cref{fig:typing-rules}.
    To improve readability, empty type contexts are omitted from rules.
\end{definition}

\begin{definition}[End Predicate]\label{def:pEnd}
    A context $\Delta$ is \tEnd-typed, by:
    {\normalfont\[
        \infer{}{end}
        {\forall i \in 1..n :  \type{\St_i} = \tEnd}
        {\pEnd(\role{p$_1$} : \type{S_1} \cdot \ldots \cdot \role{p$_n$} : \type{S_n})}
    \]}
\end{definition}

\begin{figure}[t]
    \begin{mathpar}
        \infer{T-\St}{t-s}
        { }
        {\Gamma \;; \p : \St \vdash \p : \St}

        \infer{T-Var}{t-var}
        { \Gamma(x) = \type{B} }
        {\Gamma \vdash x : \type{B}}

        \infer{T-Val}{t-val}
        { v \in \type{B} }
        {\Gamma \vdash v : \type{B}}

        \infer{T-0}{t-0}
        {\pEnd(\Delta)}
        {{\Gamma}\;; \Delta \vdash 0}

        \infer{T-\type{$\o$}}{t-send}
        { \Gamma \;; \Delta \vdash \p : \type{\o_{i \in I} \role{q$_i$} : \m_i({B_i}_1,\dots,{B_i}_n) \then S_i} \\ k \in I \\ \forall j \in 1..n : \Gamma \vdash c_j : \type{{B_k}_j} \\ {\Gamma}\;;\p : \type{S_k} \vdash P }
        {{\Gamma}\;;\Delta \vdash \p \tl \role{q$_k$} \o \m_k \langle c_1,\dots,c_n \rangle \then P} 

        \infer{T-\type{$\oldAnd$}}{t-recv}
        { \Gamma \;; \Delta \vdash \p : \type{\&_{i \in I} \role{q$_i$} : \m_i ({B_i}_1,\dots,{B_i}_n) \then S_i [, \clock \then S^\prime]} \\ \forall i \in I : \Gamma,{y_i}_1 : \type{{B_i}_1} , \dots, {y_i}_n : \type{{B_i}_n} \ ; \ \p : \type{S_i} \vdash P_i \\ [\Gamma \;; \p : \type{S^\prime} \vdash P^\prime]}
        { \Gamma \;; \Delta \vdash \p \tl \&_{i \in I} \role{q$_i$} : \m_i ({y_i}_1,\dots,{y_i}_n) \then P_i [, \clock \then P^\prime]  }

        \infer{T-\type{!}}{t-bang}
        { \Gamma(\p) = \type{{!}_{i \in I} \role{q$_i$} : \m_i ({B_i}_1,\dots,{B_i}_n) \then S_i} \\ \forall i \in I : \Gamma,{y_i}_1 : \type{{B_i}_1} , \dots, {y_i}_n : \type{{B_i}_n} \ ;\  \p : \type{S_i} \vdash P_i} 
        { \Gamma \vdash \p \tl {!}_{i \in I} \role{q$_i$} : \m_i ({y_i}_1,\dots,{y_i}_n) \then P_i } 

        \infer{T-{$\|_1$}}{t-par1}
        {\ctxs{\Gamma}{\Delta}{\emptyset} \vdash \net \\ \ctxs{\Gamma}{\emptyset}{\beta} \vdash \buf}
        {\ctxs{\Gamma}{\Delta}{\beta} \vdash \net \| \buf}

        \infer{T-{$\|_2$}}{t-par2}
        {\Gamma \,; \Deltab{1} \vdash \net_1 \\ \Gamma \;; \Deltab{2} \vdash \net_2}
        {\Gamma \,; \Deltab{1} , \Deltab{2} \vdash \net_1 \| \net_2}

        \infer{T-{$\pPar$}}{t-par3}
        {\Gamma \,; \Deltab{1} \vdash \p \tl \mathcal{P}_1 \\ \Gamma \;; \Deltab{2} \vdash \p \tl \mathcal{P}_2}
        {\Gamma \,; \Deltab{1} \cons \Deltab{2} \vdash \p \tl \mathcal{P}_1 \pPar \mathcal{P}_2}

        \infer{T-Empty}{t-empty}
        {  }
        { \ctxs{\Gamma}{\emptyset}{\beta} \vdash \emptyset}

        \infer{T-Buf}{t-buf}
        { \forall j \in 1..n : \Gamma \vdash v_j : \type{B_j} \\ \ctxs{\Gamma}{\emptyset}{\beta} \vdash \buf }
        { \ctxs{\Gamma}{\emptyset}{\beta \cons \type{(\p \rightarrow \q, \m(B_1,\dots,B_n))}} \vdash \buf \cons \langle \p \rightarrow \q, \m\langle v_1,\dots,v_n \rangle\rangle}
    \end{mathpar}
    \caption{Typing rules.}
    \label{fig:typing-rules}
\end{figure}

Typing rules \cref{rule:t-s}, \cref{rule:t-var}, \cref{rule:t-val} are auxiliary judgements typing linear roles, variables and values.
A role \p\ of type \St\ is typed by a linear context containing exactly a mapping of \p\ to \St; 
variables are typed to a basic type if that mapping is held by $\Gamma$; and
values are typed to a basic type if they belong to their sets.
The empty process $\0$ is typed by \cref{rule:t-0} if the linear context is \tEnd-typed (\cref{def:pEnd}), \ie\ $\Delta$ only contains roles mapped to \tEnd.

The send process $\p \tl \role{q$_k$} \o \m_k \langle c_1,\dots,c_n \rangle \then P$ is well typed by \cref{rule:t-send} if: 
$\Delta$ can map \p\ to a selection type containing the path chosen by the process; $\Gamma$ verifies all payloads with their types indicated in the session type; and the continuation type can check the continuation process.

The receive process $\p \tl \&_{i \in I} \role{q$_i$} : \m_i ({y_i}_1,\dots,{y_i}_n) \then P_i [, \clock \then P^\prime]$ is well typed by \cref{rule:t-recv} if:
$\Delta$ maps \p\ to a branch with all the same paths contained in $I$; the payloads and continuation types of every path in the branch can type all process continuations $P_i$; and if a timeout process $P^\prime$ is defined, then it must be typed under a timeout branch in the session type.

Replicated receive $\p \tl {!}_{i \in I} \role{q$_i$} : \m_i ({y_i}_1,\dots,{y_i}_n) \then P_i $ is typed using \cref{rule:t-bang} in a similar manner to \cref{rule:t-recv}; the type of \p\ instead lives in the unrestricted context.

Network composition is typed by \cref{rule:t-par1} and \cref{rule:t-par2}.
The former separates the linear context to be used on processes and the buffer context to be used on the network buffer; the latter splits context domains to type different roles in the network.
Process-level composition is typed via \cref{rule:t-par3} which utilises the context splitting operation (\cref{fig:context-operations}) to separate parallel session types.

Network buffers are typed by repeated applications of \cref{rule:t-buf}, which removes messages from the buffer one at a time if they match a message type in the type buffer.
The empty buffer is typed under \cref{rule:t-empty}, allowing for possible leftover types in $\Theta$. 
It is key to note that the buffer context is \emph{affine}, as any message that gets dropped at runtime will result in an unused message type.

\begin{example}[Interactions with Failure: Types]\label{ex:type-failures}
    Due to the generalised nature of the type system, the type judgement alone is not enough to detect the errors that may occur in $\net_f$.
    This is because the type-system does not provide \emph{syntactic} guarantees, but rather should be used in conjunction with exhaustive verification techniques post protocol design (this is standard in generalised MPST~\cite{DBLP:journals/pacmpl/ScalasY19,DBLP:conf/esop/BrunD23,DBLP:conf/concur/BarwellSY022}).
    In fact, network $\net_f$ can be typed under the following contexts:
    \[
        \ctxs{\Gamma}{\p : \type{\o\, \q : \m(\mathbb{N}) \then S}, \q : \type{\& \{\p : \m(\basic{String}) \then S^\prime,\ \clock \then S^{\prime\prime}\}}}{\Theta \cons \type{(\p \rightarrow \q, \m(\basic{String}))}}
    \]

    \noindent
    for some $\Gamma, \Theta, \type{S}, \type{S^\prime}, \type{S^{\prime\prime}}$ assuming that $P$, $P^\prime$ and $P^{\prime\prime}$ are well typed using $\type{S}$, $\type{S^\prime}$ and $\type{S^{\prime\prime}}$ respectively.
    Note that $\Gamma$ and $\Theta$ can be non-empty since the former is unrestricted and the latter is affine. 
    In contrast, the linear context must be exactly as stated above.
    We now need a way to determine this protocol as unsafe. 
\end{example}

\section{Songs About Songs}\label{sec:meta}

Unlike most session type theories, \emph{generalised} MPST do not syntactically guarantee any properties on the processes they type.
Rather, they provide a framework for \emph{exhaustively checking} runtime properties on the type context, from which process-level properties may be inferred.
This seemingly unconventional approach to session types was discovered to be \emph{more expressive} than its syntactic counterpart \wrt\ the amount of well-typed programs it can capture~\cite{DBLP:journals/pacmpl/ScalasY19}.
Furthermore, its generalised nature allows for fine-tuning based on specific requirements of its applications.
Informally, generalisation of the type system works by proving the meta-theory parametric of a safety property; \ie\ all theorems proved and presented assume that the type contexts are \emph{safe} (\cref{sec:safety}).
With this assumption we present our main results in \cref{sec:properties}.

\subsection{Type Safety}\label{sec:safety}

The technical definition of \emph{safety} refers to the \emph{minimal requirements} on types to guarantee \emph{subjection reduction} (\cf~\cref{sec:properties}, \cref{thm:sr}). 
But what does safety even mean for a distributed network with message loss, delays and reordering?
It is impossible for our type system to adopt standard notions of safety which may guarantee properties such as \emph{no unexpected messages} or \emph{correct ordering of messages}, since the failures experienced at runtime can mitigate such guarantees.
Hence, the minimal guarantee of safety (\cref{def:prop-safe}) in \magpi{!} ensures that: 
\begin{enumerate}
    \item timeout branches are always (and only) defined for failure-prone communication between \emph{linear} processes; and
    \item if a message eventually reaches its destination, then the expected types of the payload from the recipient should match the data carried on the message.
\end{enumerate}

\begin{definition}[Safety Property]\label{def:prop-safe}
    $\safe_\R$ is a safety property on contexts iff:
    {\normalfont\begin{mathpar}\mprset{flushleft}
        \infer{$\safe$-r$_1$}{safe-r1}{}
        {\safe_\R(\ctxs{\Gamma}{\Delta \cons \p : \type{\&_{i \in I} \role{q$_i$} : \m_i (\tilde{B_i}) \then S_i}}{\beta}) \textit{  implies  } \forall i \in I : \{\role{q$_i$}, \p\} \in \R  }

        \infer{$\safe$-r$_2$}{safe-r2}{}
        {\safe_\R(\ctxs{\Gamma}{\Delta \cons \p : \type{\&_{i \in I} \role{q$_i$} : \m_i (\tilde{B_i}) \then S_i, \clock \then S^\prime}}{\beta}) \textit{  implies  } \exists k \in I : \{\role{q$_k$}, \p\} \not\in \R  }
        
        \infer{$\safe$-c}{safe-c}{}
        {\safe_\R(\ctxs{\Gamma}{\Delta \cons \p : \type{\&_{i \in I} \role{q$_i$} : \m_i (\tilde{B_i}) \then S_i [, \clock \then S^\prime]}}{\beta \cons \type{(\role{q$_k$} \rightarrow \p, \m_k(\tilde{B^\prime}))}}) \\\\ \tab\tab\textit{and } k \in I \textit{  implies  } |\type{\tilde{B_k}}| = |\type{\tilde{B^\prime}}| \textit{ and } \forall j \in 1..|\type{\tilde{B_k}}| : \type{{B_k}_j} = \type{{B^\prime}_j} }

        \infer{$\safe$-{!}c}{safe-!c}{}
        {\safe_\R(\ctxs{\Gamma, \p : \type{{!}_{i \in I} \role{q$_i$} : \m_i (\tilde{B_i}) \then S_i}}{\Delta}{\beta \cons \type{(\role{q$_k$} \rightarrow \p, \m_k(\tilde{B^\prime}))}}) \\\\ \tab\tab\textit{and } k \in I \textit{  implies  } |\type{\tilde{B_k}}| = |\type{\tilde{B^\prime}}| \textit{ and } \forall j \in 1..|\type{\tilde{B_k}}| : \type{{B_k}_j} = \type{{B^\prime}_j} }

        \infer{$\safe$-$\redP$}{safe-r}{}
        {\forall \Deltap : \safe_\R(\ctxs{\Gamma}{\Delta}{\beta}) \,\textit{and}\, \ctxs{\Gamma}{\Delta}{\beta} \redP \ctxs{\Gamma}{\Deltap}{\betap} \textit{  implies  } \safe_\R(\ctxs{\Gamma}{\Deltap}{\betap})}
    \end{mathpar}}
\end{definition}

Conditions \cref{rule:safe-r1} and \cref{rule:safe-r2} ensure that timeouts are \emph{only} omitted (resp. defined) when communication is reliable (resp. unreliable).
\cref{rule:safe-c} and \cref{rule:safe-!c} require payload types to match for any communication; note that no message is ever incorrectly delivered to a linear channel instead of a replicated (and \vv) because we assume that message labels for replicated receives are not reused in their continuations.
The last condition, \cref{rule:safe-r}, requires \emph{all} possible reductions of safe contexts to also be safe.

\begin{example}[Interactions with Failure: Safety]\label{ex:safety-failures}
    The type contexts presented in \cref{ex:type-failures} do not abide by the conditions of $\safe_\emptyset$ and thus are not safe.
    The types do meet conditions \cref{rule:safe-r1} to \cref{rule:safe-!c}, but fail \cref{rule:safe-r}.
    We observe the following traces of the LTS:

    
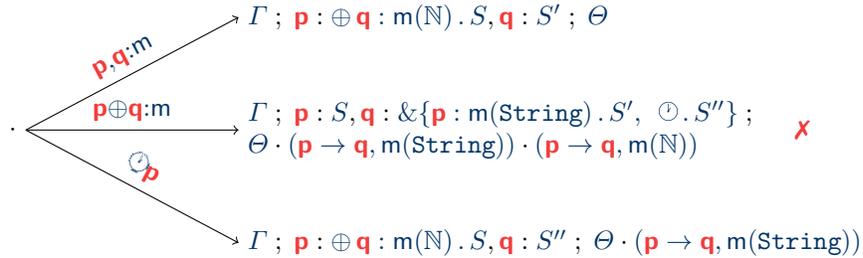
\begin{tikzpicture}
        \node (s)  at (0,0)                        {$\cdot$};
        \node (r1) at (3,1.5)  [anchor=west] {$\ctxs{\Gamma}{\p : \type{\o\, \q : \m(\mathbb{N}) \then S}, \q : \type{S^\prime}}{\Theta}$};
        \node (r2) at (3,0)    [anchor=west, text width=7.2cm] {$\Gamma \;;\; \p : \type{S}, \q : \type{\& \{\p : \m(\basic{String}) \then S^\prime,\ \clock \then S^{\prime\prime}\}}\;; $\\$ \Theta \cons \type{(\p \rightarrow \q, \m(\basic{String}))} \cons \type{(\p \rightarrow \q, \m(\mathbb{N}))}$};
        \node (r3) at (3,-1.5) [anchor=west] {$\ctxs{\Gamma}{\p : \type{\o\, \q : \m(\mathbb{N}) \then S}, \q : \type{S^{\prime\prime}}}{\Theta \cons \type{(\p \rightarrow \q, \m(\basic{String}))}}$};
        \draw[->] (s.east) -- (r1.west) node[midway,above,sloped] {$\aCom{\role{\p}}{\q}{\m}$};
        \draw[->] (s.east) -- (r2.west) node[midway,above,sloped] {$\aOut{\p}{\q}{\m}$};
        \draw[->] (s.east) -- (r3.west) node[midway,above,sloped] {$\aTime{p}$};
        \node (x) at (10.5,0) {\xmark};
    \end{tikzpicture}

    \noindent
    The transition over label $\aOut{\p}{\q}{\m}$ yields contexts in violation of \cref{rule:safe-c}.
    This example highlights the impact of message labels in protocol design, as reusing labels may lead to nondeterministic receipt of messages.
    However, this does not mean that messages with the same label can never be reused---it is possible for this nondeterminism to still be safe w.r.t. \cref{def:prop-safe}. 
    \Eg\ consider the types in \cref{ex:motivation}, the types reuse labels \ping\ and \pong.
    This is safe because the protocol has no dependency on receiving messages with the same label in a specific order.
\end{example}

\begin{example}[Load Balancer: Types]
    We type our load balancer using the protocol below in a judgement as $\ctxs{\role{s} \!:\! \Rt_\role{s}, \role{w$_1$} \!:\! \Rt_\role{w$_1$}, \role{w$_2$} \!:\! \Rt_\role{w$_2$}}{\role{c} \!:\! \St_\role{c}, \role{o} \!:\! \St_\role{o}}{\emptyset} \vdash \net \| \emptyset$ where $\net$ contains the processes from \cref{ex:processes}.
    The protocol observes the safety property \wrt\ the reliability relation defined in \cref{ex:network}.

    {\small \begin{align*}
        \Rt_\role{s} &= \type{{!}\role{c} : \msgLabel{req}(\mathbb{N}) \then \o \left\{\begin{array}{l}
            \role{w$_1$} : \msgLabel{req}(\mathbb{N}) \then \tEnd \\
            \role{w$_2$} : \msgLabel{req}(\mathbb{N}) \then \tEnd
        \end{array} \right.}  \\
        \Rt_\role{w$_1$} &= \type{{!} \role{s} : \msgLabel{req}(\mathbb{N}) \then \o \role{c} : \msgLabel{ans}(\texttt{Real}) \then \tEnd}\\
        \Rt_\role{w$_2$} &= \type{{!} \role{s} : \msgLabel{req}(\mathbb{N}) \then \o \role{c} : \msgLabel{ans}(\texttt{Real}) \then \tEnd}\\
        \St_\role{c} &= \type{\o \role{s} : \msgLabel{req}(\mathbb{N}) \then \& \left\{\begin{array}{l}
            \role{w$_1$} : \msgLabel{ans}(\texttt{Real}) \then \o \role{o} : \msgLabel{output}(\texttt{Real}) \then \tEnd \\
            \role{w$_2$} : \msgLabel{ans}(\texttt{Real}) \then \o \role{o} : \msgLabel{output}(\texttt{Real}) \then \tEnd \\
            \clock \then \o \role{o} : \msgLabel{err}(\texttt{String}) \then \tEnd
        \end{array} \right.} \\
        \St_\role{o} &= \type{\& \{\role{c} : \msgLabel{output}(\texttt{Real}) \then \tEnd, \role{c} : \msgLabel{err}(\texttt{String}) \then \tEnd\}}
    \end{align*}}
\end{example}

\subsection{Type Properties}\label{sec:properties}

Our main results are presented below (proof details are given in \cref{app:proofs}).
\emph{Subject reduction} (\cref{thm:sr}) states that any process typed under a safe context remains well-typed and safe after reduction (even in the presence of failures).
From this we obtain \cref{cor:fhg}, stating that timeout branches are only omitted from linear receives if communication is reliable; hence certifying that all processes typed by safe contexts guarantee that no \emph{linear} failure-prone communication goes unhandled.
A key contribution of our work is that this corollary is relaxed to \emph{linear} processes instead of \emph{all} processes, since we do not wish for replicated servers to handle dropped client requests.

\begin{theorem}[Subject Reduction]\label{thm:sr}
    If $\ctxs{\Gamma}{\Delta}{\Theta} \vdash \net$ with $\safe_\R(\ctxs{\Gamma}{\Delta}{\Theta})$ and $\net \redP_\R \net^\prime$, then $\exists \Deltap,\Thetap$ s.t. $\ctxs{\Gamma}{\Delta}{\Theta}\redP^*\ctxs{\Gamma}{\Deltap}{\Thetap}$ and $\ctxs{\Gamma}{\Deltap}{\Thetap} \vdash \net^\prime$ with $\safe_\R(\ctxs{\Gamma}{\Deltap}{\Thetap})$. 
\end{theorem}
\begin{corollary}[Failure Handling Guarantee]\label{cor:fhg}
    If $\ctxs{\Gamma}{\Delta}{\Theta} \vdash \net$ with $\safe_\R(\ctxs{\Gamma}{\Delta}{\Theta})$ and $\net \redP_\R^* \p \tl \&_{i \in I} \role{q$_i$} : \m_i (\tilde{c_i}) \then P_i \pPar \mathcal{P} \| \net^{\prime}$, then $\forall i \in I : \{\p, \role{q$_i$}\} \in \R$.
\end{corollary}
\emph{Session fidelity} (\cref{thm:sf}) states the opposite implication \wrt\ subjection reduction, \ie\ processes typed under a safe context can always match \emph{at least one} reduction available to the context.
\begin{theorem}[Session Fidelity]\label{thm:sf}
    If $\ctxs{\Gamma}{\Delta}{\Theta} \redP$ and $\ctxs{\Gamma}{\Delta}{\Theta} \vdash \net$ with $\safe_\R(\ctxs{\Gamma}{\Delta}{\Theta})$, then $ \exists \Deltap,\Thetap,\net^\prime$ s.t. $\ctxs{\Gamma}{\Delta}{\Theta} \redP \ctxs{\Gamma}{\Deltap}{\Thetap}$ and $\net \redP_\R^* \net^\prime$ and $\ctxs{\Gamma}{\Deltap}{\Thetap} \vdash \net^\prime$ with $\safe_\R(\ctxs{\Gamma}{\Deltap}{\Thetap})$. 
\end{theorem}
Using this result we can verify properties other than just safety.
This is the benefit of the generalised approach to MPST, where instead of forcing protocols to abide by specific properties, types can be checked \emph{a posteriori} to determine any properties they observe.
We demonstrate for \emph{deadlock freedom} (\cref{def:df-nets}).

\begin{definition}[DF: Networks]\label{def:df-nets}
    A network $\net$ is deadlock free, written \textnormal{$\df(\net)$}, iff $\net \redP^* \net^\prime \not\redP$ implies either \begin{enumerate}
        \item $\net^\prime \equiv \0 \| \buf$; or
        \item $\net^\prime \equiv\ \net_1^{\prime} \| \cdots \| \net_n^\prime \| \buf$ s.t. \textnormal{$\forall i \in 1..n : \net_i^{\prime} = \role{p$_i$} \tl {!}_{j \in J} \role{q$_j$} : \m_j(\tilde{x_j}) \then P_j$}.
    \end{enumerate}
\end{definition}
A deadlock free network is one that only gets stuck when all processes reach $\0$, or when the only non-$\0$ processes left in the network are servers.
(Note, the buffer is allowed to be non-empty because of message delays.)
We define deadlock freedom on types in \cref{def:df-types}, stating that type contexts are deadlock free if they only get stuck when the linear context is \tEnd-typed.

\begin{definition}[DF: Types]\label{def:df-types}
    Contexts $\ctxs{\Gamma}{\Delta}{\Theta}$ are deadlock free, written \textnormal{$\df(\ctxs{\Gamma}{\Delta}{\Theta})$}, iff $\ctxs{\Gamma}{\Delta}{\Theta} \redP^* \ctxs{\Gamma}{\Deltap}{\Thetap} \not\redP$ implies \textnormal{$\pEnd(\Deltap)$}.
\end{definition}
\begin{proposition}[Property Verification: DF]\label{prop:df-ver}
    If $\ctxs{\Gamma}{\Delta}{\Theta} \vdash \net$ with $\safe_\R(\ctxs{\Gamma}{\Delta}{\Theta})$, then $\df(\ctxs{\Gamma}{\Delta}{\Theta})$ implies $\df(\net)$.
\end{proposition}
Lastly, in \cref{prop:df-ver} we state that deadlock free contexts imply deadlock freedom in the networks they type, a result which follows from \cref{thm:sf}.

\subsubsection{Decidability.}\label{sec:decidable}

Asynchronous generalised MPST are known to be undecidable in general~\cite{DBLP:journals/pacmpl/ScalasY19,DBLP:conf/esop/BrunD23}.
This stems from the fact that session types with asynchronous buffers can encode Turing machines~\cite[Theorem 2.5]{DBLP:journals/corr/abs-1211-2609}.
However, we note that this simulation relies on buffers with queue semantics and tail-recursion; whereas our type system uses bag buffers and replication.
Comparing the expressive power of recursion and replication, previous studies show that for $\pi$-calculi with communication of free names the two are equally as expressive~\cite{DBLP:journals/eatcs/Palamidessi05}; whereas without communication of free names (\eg\ CCS) recursion is strictly more expressive than replication~\cite{DBLP:journals/mscs/BusiGZ09}.
Thus, we raise the question: ``\textit{What is the expressive power of asynchronous session types with bag buffers and replication?}'', which we aim to answer in future work.

For now, we present a predicate on type contexts which can be used to determine decidable subsets of the type system.
This predicate, which we call \emph{trivially terminating} (\cref{def:tt}), is decidable and implies decidability of safety (and subsequently property verification).

\begin{definition}[Trivially Terminating]\label{def:tt}
    We say $\ctxs{\Gamma}{\Delta}{\Theta}$ are {trivially terminating}, written \textnormal{$\tterm(\ctxs{\Gamma}{\Delta}{\Theta})$}, iff \textnormal{$\forall \p \in \dom(\Gamma): \Gamma(\p) = \type{{!}_{i \in I} \role{q$_i$} : \m_i (\tilde{B}_i) \then \tilde{S}_i}$ where $\forall i \in I : \role{q$_i$} \not\in \dom(\Gamma)$}.
\end{definition}

\begin{proposition}[Decidable Subset]\label{prop:decidable}
    For any contexts, \textnormal{$\tterm(\ctxs{\Gamma}{\Delta}{\Theta})$} is decidable and \textnormal{$\tterm(\ctxs{\Gamma}{\Delta}{\Theta})$} implies checking \textnormal{$\safe_\R(\ctxs{\Gamma}{\Delta}{\Theta})$} is decidable.
\end{proposition}

\section{Encore}\label{sec:conc}

Modelling of failures and distributed communication is increasingly becoming a more relevant and widely researched topic within the area of programming languages.
We highlight below some key related work, identifying the main differences \wrt\ \magpi{!}.

\emph{Affine session types}~\cite{DBLP:journals/lmcs/MostrousV18,DBLP:journals/pacmpl/FowlerLMD19,DBLP:journals/mscs/CapecchiGY16} use affine typing to allow sessions to be prematurely cancelled in the event of failure.
They may be used in a similar fashion to try-catch blocks, where a main protocol is followed until a possible failure is met and handled gracefully.
Similar in approach to \magpi{!}\ is work by Barwell~\etal~\cite{DBLP:conf/concur/BarwellSY022}, where generalised MPST theory is extended to reason about \emph{crash-stop failures}.
Where \magpi{!} uses timeouts, the previous uses a ``crash'' message label which can be fed to a receiving process via some assumed failure detection mechanism.
Viering~\etal~\cite{DBLP:journals/pacmpl/VieringHEZ21} present an \emph{event-driven and distributed} MPST theory, where a central robust node is assumed and is capable of restarting crashed processes.
Chen~\etal~\cite{DBLP:conf/forte/ChenVBZE16} remove the dependency on a reliable node, instead using \emph{synchronisation points} to handle failures as they are detected.
Adameit~\etal~\cite{DBLP:conf/forte/AdameitPN17} consider session types for \emph{link failures} where \emph{default values} act as failure-handling mechanism to substitute lost data.
\magpi{!} models lower-level failures than all of these works.
Most of the aforementioned assume some perfect failure-detection mechanism, whereas \magpi{!} embraces timeouts as a weak failure detector to show that some degree of safety can still be achieved.
Our theory is designed to operate at a lower level of abstraction, thus often providing weaker guarantees (\eg\ consider our minimal definition of type safety) in exchange for modeling a wider set of communication failures.

The adoption of replication in MPST theory is a novel contribution of this paper.
Replication in broader session types research has been utilised on numerous accounts~\cite{DBLP:journals/mscs/CairesPT16,DBLP:conf/fossacs/DardhaG18,DBLP:conf/esop/CairesP17,DBLP:journals/pacmpl/QianKB21},
specifically in work pertaining to Curry-Howard interpretations of linear logic as session types, where the exponential modality from linear logic $!A$ is typically linked to replication from the $\pi$-calculus.  
Disregarding our modeling of failure, the largest difference between these works and ours is that we focus on a multiparty setting, whereas these theories are all based on binary communication.
Furthermore, we did not opt to approach our problem from a logic-perspective, as is the main motivation behind this line of research.
Instead, we build upon already-standard generalised MPST theory, adapting it towards our problem domain.
We do note, however, that exploring a logical approach to replication in MPST (and, in turn, to failures in session types) is an interesting direction for future work.
A more related use of replication in types is by Marshall and Orchard~\cite{DBLP:journals/corr/abs-2203-12875}, where the authors discuss how non-linear types can be used in a controlled fashion to type behaviours such as repeatedly spawning processes.
This resembles the semantics of our type system and dynamic definition of replication in our language, where replicated processes (resp. types) can be reused as necessary to pull out linear copies of their continuations.
The mentioned work focuses on how to control the use of non-linear types and how this can be utilised with session types in a functional programming language.
Our work, on the other hand, uses replication as a means of better modeling client-server interactions and distinguishing between failure-prone communication that should be handled by the recipient or the sender.

On session types for client-server communication, research largely takes the approach of linear-logic correspondences~\cite{DBLP:journals/pacmpl/QianKB21,DBLP:conf/types/Padovani22,DBLP:conf/esop/CairesP17,DBLP:journals/jfp/Wadler14}.
The topology these works target are of binary sessions between a pool of clients and a single server.
In Qian et al. \cite{DBLP:journals/pacmpl/QianKB21} a logic is developed, called \textsf{CSLL} (client-server linear logic), utilising the \emph{coexponential} modality \textexclamdown$A$.
The subtle difference between this modality and the exponential $!A$ is that the latter represents an unlimited number of a type $A$, while the former serves type $A$ as many times as required according to client requests, in a sequential yet still unordered manner.
This is very similar to how our type system operates, given that replicated receives only pull out copies of continuations upon communication. 
Multiple requests induce non-determinism into further reductions, in our work this is seen in the extension of parallel types, which in Qian et al. \cite{DBLP:journals/pacmpl/QianKB21} is observed through hyper-environments.
The difference in goal between the work of \cite{DBLP:journals/pacmpl/QianKB21} (followed up by \cite{DBLP:conf/types/Padovani22}) and ours, is that the mentioned works focus on providing fixed static guarantees on the processes they type (the former work with a focus on deadlock freedom, the latter on weak termination) whilst we take a generalised approach.
Our type system does not force programs to be deadlock-free or terminating, but rather requires a less restrictive \emph{safety} property and allows verification of deadlock freedom and termination to be done \emph{a posteriori}---the trade-off being our weaker form of type safety given the failure-prone nature of our setting.

To conclude, we presented \magpi{!}, an extension to \magpi\ made to use replication (instead of recursion) to express infinite computation---both at the language and type levels.
We did so with the aim of better modelling multiparty client-server interactions, where servers are designed to remain infinitely available.
Specifically, we find type-level replication to be a clean mechanism for offloading the handling of certain failures from the recipient to the sender---a practical procedure for client-server interactions.
We have generalised our theory by proving our meta-theoretic results parametric of the largest safety property, allowing for more specific properties to be instantiated and used to verify runtime behaviours.
As future work, we plan to investigate more specific properties for verification through our general type system.
We aim to explore in detail the decidability of type-level properties and if/how they may be restricted to obtain decidable bounds in cases where they are not.
Lastly, we wish to conduct a foundational study of the use of replication in MPST---we anticipate their use for modelling client-server interactions to have further benefit outwith a failure-prone setting.

\subsubsection*{Acknowledgements.}
Supported by the UK EPSRC New Investigator Award grant EP/X027309/1 ``Uni-pi: safety, adaptability and resilience in distributed ecosystems, by construction''.

\clearpage

%
%
%
\bibliographystyle{splncs04}
\bibliography{mybibliography}

\clearpage

\changetext{}{10em}{-5em}{-5em}{} 
\appendix
\section{Simplification via Replication}\label{app:simplify}

We compare session types for the \emph{ping} protocol written in \magpi~\cite[Ex. 1]{DBLP:conf/esop/BrunD23} with our extension, \magpi{!}.
The example consists of three participants (or \emph{roles}): the client \p,  server \q, and a result channel \rr.
Client side communication between \p\ and \rr\ is reliable, as we assume these are processes on a single machine.
Communication with the server \q\ is \emph{unreliable}, thus failures must be handled through \emph{timeouts}.
The client-side session types are defined as $\St_\rr$ and $\St_\p$ below.
{\color{type} \[\begin{array}{l l}
        \St_\rr &= \&\{\p : \ok \then \tEnd,\ \p : \ko \then \tEnd\} \\
        \St_\p &= \o\, \q : \ping \then\!\! \&\!\! \left\{\begin{array}{l}
            \q : \pong \then \o \rr : \ok \then \tEnd, \\
            \clock \then\! \o \q : \ping \then\!\! \&\!\! \left\{\begin{array}{l}
                {\q : \pong \then \o \rr : \ok \then \tEnd}, \\
                \clock \then\! \o \q : \ping \then\!\! \&\!\! \left\{\begin{array}{l}
                    {\q : \pong \then\!\! \o \rr : \ok \then \tEnd}, \\
                    {\clock \then\! \o \rr : \ko \then \tEnd}
                \end{array}\right.
            \end{array}\right.
        \end{array}\right.
    \end{array}
\]}

\noindent
Client \p\ begins by sending a message with label \ping\ to server \q\ (\type{$\o\, \q : \ping$}) and then waits for a \pong\ response (\type{$\& \q : \pong$}).
If successful, an \ok\ message is sent to the results role \rr\ and the session is terminated for the client (\tEnd).
Since communication with the server is \emph{unreliable}, receipt of the \pong\ message is not guaranteed, and thus must be handled by a \emph{timeout} branch \type{$\clock$}.
The client attempts to reach the server 3 times---if all attempts fail, it sends a \ko\ message to \rr.
The result role \rr\ waits for either of the reliable responses from \p, thus no timeout is defined.

Now we consider the server-side protocol. 
One possible definition of the server session type could be $\St_\q$ below (as show in \cite[Ex. 1]{DBLP:conf/esop/BrunD23}).

{\color{type} \[
    \type{\St_\q} = \type{\&} \left\{\begin{array}{l}
        \p : \ping \then \o \p : \pong \then \tEnd \\
        \clock \then \type{\&} \left\{\begin{array}{l}
            \p : \ping \then \o \p : \pong \then \tEnd \\
            \clock \then \type{\&} \left\{\begin{array}{l}
                \p : \ping \then \o \p : \pong \then \tEnd\\
                \clock \then \tEnd
            \end{array}\right.
        \end{array}\right.
    \end{array}\right.
\]}

\noindent
However, although this is a \emph{safe} definition, it does not cater for \emph{all} possible scenarios.
The above definition only allows for up to three failures of the initial \ping, and not the \pong\ reply.
If any \pong\ reply were to fail, then the result would always be \ko, even though the client may try again.
If we were to also cater for three attempts from the server-side, then the server session type in \magpi\ could be defined as $\type{S_\q^\prime}$.

{\small\color{type} \[
    \type{\St_\q^\prime} = \type{\&} \left\{\begin{array}{l}
        \p : \ping \then \o \p : \pong \then \type{\&} \left\{\begin{array}{l}
            \p : \ping \then \o \p : \pong \then \type{\&} \left\{\begin{array}{l}
                \p : \ping \then \o \p : \pong \then \tEnd \\
                \clock \then \tEnd
            \end{array}\right.\\
            \clock \then \type{\&} \left\{\begin{array}{l}
                \p : \ping \then \o \p : \pong \then \tEnd \\
                \clock \then \tEnd
            \end{array}\right.
        \end{array}\right.\\
        \clock \then \type{\&} \left\{\begin{array}{l}
            \p : \ping \then \o \p : \pong \then \type{\&} \left\{\begin{array}{l}
                \p : \ping \then \o \p : \pong \then \tEnd\\
                \clock \then \tEnd
            \end{array}\right.\\
            \clock \then \type{\&} \left\{\begin{array}{l}
                \p : \ping \then \o \p : \pong \then \tEnd\\
                \clock \then \tEnd
            \end{array}\right.
        \end{array}\right.
    \end{array}\right.
\]}

The key difference here is that the server must listen for a \ping\ even after its \pong\ reply, since \pong\ could fail and the client will reattempt the request.
In practice, servers avoid these cumbersome definitions by doing away with a specific retry count.
Retries (and the handling of failures in general), is left up to the client; servers instead opt to listen for requests for as long as they are available.
Hence, no matter the direction in which a message is dropped, servers will always attempt to respond to requests (unlike $\St_\q$).
We demonstrate this approach using replication at type-level in \magpi{!}. 
\[\type{\St_\q^{!} = {!}\p : \ping \then \o \p : \pong \then \tEnd}\]

\noindent
The replicated receive \type{${!}\p : \ping$} denotes the server's constant availability to receive a \ping\ request.
Upon receiving a message from the client, the \pong\ response is issued, but the server \emph{remains available} to receive further requests.
We highlight the absence of a failure-handling timeout branch in \type{$\St_\q^{!}$}; the server does not need to change its behaviour if a client request fails.
Furthermore, if the \pong\ reply fails, the server remains available to handle any number of retries from the client.
Thus, $\St_\q^{!}$ handles all the failures covered by \type{$\St_\q^\prime$} whilst also being agnostic to the client-side implementation (\ie\ can handle any number of retries) and reducing complexity of the server-side protocol.

\section{Structural Congruence}\label{app:cong}

The full list of congruence rules used in \magpi{!} is given in \cref{def:congruence} below.
Note that, different to most standards, no congruence rules are used for replication---extracting processes out of replicated receives is treated as an irreversible action and is handled by reduction (\cref{fig:semantics}).

\begin{definition}[Structural Congruence]\label{def:congruence}
    \begin{mathpar}
    \infer{}{cong-par-sym}
        {}
        {\net_1 \| \net_2 \equiv \net_2 \| \net_1}

    \infer{}{cong-par-ass}
        {}
        {(\net_1 \| \net_2) \| \net_3 \equiv \net_1 \| (\net_2 \| \net_3)}

    \infer[\textnormal{if \p\ $\not \in \roles(\net)$}]{}{cong-par-id}
        {}
        {\net \| \p \tl \0 \equiv \net}

    \infer{}{proc-par-sym}
        {}
        {\mathcal{P}_1 \pPar \mathcal{P}_2 \equiv \mathcal{P}_2 \pPar \mathcal{P}_1}

    \infer{}{proc-par-ass}
        {}
        {(\mathcal{P}_1 \pPar \mathcal{P}_2) \pPar \mathcal{P}_3 \equiv \mathcal{P}_1 \pPar (\mathcal{P}_2 \pPar \mathcal{P}_3)}

    \infer{}{proc-par-id}
        {}
        {\mathcal{P} \pPar \0 \equiv \mathcal{P}}
\end{mathpar}
\end{definition}

\section{Subject Reduction and Session Fidelity}\label{app:proofs}

\begin{customthm}{1}[Subject Reduction]
    If $\ctxs{\Gamma}{\Delta}{\Theta} \vdash \net$ with $\safe_\R(\ctxs{\Gamma}{\Delta}{\Theta})$ and $\net \redP_\R \net^\prime$, then $\exists \Deltap,\Thetap$ s.t. $\ctxs{\Gamma}{\Delta}{\Theta}\redP^*\ctxs{\Gamma}{\Deltap}{\Thetap}$ and $\ctxs{\Gamma}{\Deltap}{\Thetap} \vdash \net^\prime$ with $\safe_\R(\ctxs{\Gamma}{\Deltap}{\Thetap})$. 
\end{customthm}

\begin{proof}[Sketch]
    The proof is by induction on the derivation of $\net \redP_\R \net^\prime$.
    For most cases, we infer the shape of relevant types in the contexts by inversion of typing rules and the assumption of $\ctxs{\Gamma}{\Delta}{\Theta} \vdash \net$.
    Any message in the network buffer being used for reduction (by sending or receiving) is present in the type buffer---recall the affine nature of type buffers, and the fact that failures only drop network level messages.
    Therefore, the contexts can match the network reduction, and we already know that the reduced contexts can type the process continuations from the typing rules.
    The last thing to show is that the reduced contexts are also safe, which is obtained via $0$ or more applications of \cref{rule:safe-r}.

    For the case where $\net$ reduces via \cref{rule:f-drop}, the network reduction is typed under the same context as the assumption, because $\Theta$ is affine and will allow for an unused message type in the type buffer. \qed
\end{proof}

\begin{customthm}{2}[Session Fidelity]
    If $\ctxs{\Gamma}{\Delta}{\Theta} \redP$ and $\ctxs{\Gamma}{\Delta}{\Theta} \vdash \net$ with $\safe_\R(\ctxs{\Gamma}{\Delta}{\Theta})$, then $ \exists \Deltap,\Thetap,\net^\prime$ s.t. $\ctxs{\Gamma}{\Delta}{\Theta} \redP \ctxs{\Gamma}{\Deltap}{\Thetap}$ and $\net \redP_\R^* \net^\prime$ and $\ctxs{\Gamma}{\Deltap}{\Thetap} \vdash \net^\prime$ with $\safe_\R(\ctxs{\Gamma}{\Deltap}{\Thetap})$. 
\end{customthm}

\begin{proof}[Sketch]
    The proof begins by induction on the derivation of $\ctxs{\Gamma}{\Delta}{\Theta} \redP$, for which we observe 4 possible cases from \cref{fig:type-lts}.
    Case \cref{rule:D-time} is straight forward. From the assumption $\ctxs{\Gamma}{\Delta}{\Theta} \vdash \net$ and rule \cref{rule:D-time}, using action $\redS{\aTime{p}}$ we may infer the shape of $\net$ to \emph{at least} contain process \p\ with a top level linear receive with a defined timeout branch.
    Therefore the context action can be matched by the network via an application of \cref{rule:f-timeout}.

    For the sending case, we infer the shape $\net$ by inversion of the typing rule \cref{rule:t-send}.
    From this we obtain a network that can match \emph{exactly one} of the paths identified in the type-level selection.
    Receiving cases are similar to each other.
    By $\safe$ we know that any receive without a timeout is reliable, thus the network is guaranteed to not drop the messages required to match the type-level reduction.
    If the receive is unreliable, then the process can always, at the very least, match the timeout action.
    The remaining case is an unreliable communication to a replicated receive, as there is no guarantee that the message used for type-level reduction still exists in the network.
    In this case, if the message towards the replicated receive has already failed, then by $\safe$, either the continuation of the replicated receive is \tEnd-typed; or there is at least one other reduction possible by the contexts which the network can match. \qed
\end{proof}

\section{Property Verification}\label{app:verification}

To demonstrate how our generalised meta-theory can verify configurable runtime properties, we 
\begin{enumerate*}[label=\textit{(\roman*)}]
    \item define two sample properties on networks (deadlock freedom and termination); 
    \item present their equivalent on type contexts; and 
    \item prove the context property implies the network property.
\end{enumerate*}

\begin{definition}[Network Properties]\label{def:proc-props}
    \begin{enumerate}
        \item A network $\net$ is deadlock free, written \textnormal{$\df(\net)$}, iff $\net \redP^* \net^\prime \not\redP$ implies either \begin{enumerate}
            \item $\net^\prime \equiv 0 \| \buf$; or
            \item $\net^\prime \equiv\ \net_1^{\prime} \| \cdots \| \net_n^\prime \| \buf$ s.t. \textnormal{$\forall i \in 1..n : \net_i^{\prime} = \role{p$_i$} \tl {!}_{j \in J} \role{q$_j$} : \m_j(\tilde{x_j}) \then P_j$}.
        \end{enumerate}
        
        \item A network $\net$ is terminating, written \textnormal{$\term(\net)$}, iff \textnormal{$\df(\net)$} and $\exists k$ finite s.t. $\forall n \geq k : \net = \net_0 \redP \net_1 \redP \cdots \redP \net_n$ either \begin{enumerate}
            \item $\net_n \equiv 0 \| \buf$; or
            \item $\net^\prime \equiv\ \net_1^{\prime} \| \cdots \| \net_n^\prime \| \buf$ s.t. \textnormal{$\forall i \in 1..n : \net_i^{\prime} = \role{p$_i$} \tl {!}_{j \in J} \role{q$_j$} : \m_j(\tilde{x_j}) \then P_j$}.
        \end{enumerate}
    \end{enumerate}
\end{definition}

Note that our definitions of deadlock freedom and termination (\cref{def:proc-props}) differ from standard definitions in the $\pi$-calculus.
Specifically, it is typical for a process to only be considered deadlock free if it only stops reducing because of reaching the inactive process.
Our definition, on the other hand, allows for \emph{leftover replicated receives}.
In the domain of client-server interactions where servers are designed to be infinitely available, a server will never reach $\0$.
This is an intended design choice of the calculus, as replicated receives are meant to model servers which \emph{do not terminate} and may occasionally \emph{get stuck} waiting in an idle state.
\Cref{def:proc-props} reflects this design choice.

\begin{definition}[Type Properties]\label{def:type-props}
    \begin{enumerate}
        \item Contexts $\ctxs{\Gamma}{\Delta}{\Theta}$ are deadlock free, written \textnormal{$\df(\ctxs{\Gamma}{\Delta}{\Theta})$}, iff $\ctxs{\Gamma}{\Delta}{\Theta} \redP^* \ctxs{\Gamma}{\Deltap}{\Thetap} \not\redP$ implies \textnormal{$\pEnd(\Deltap)$}.
        \item A context $\ctxs{\Gamma}{\Deltab{0}}{\Thetab{0}}$ is terminating, written \textnormal{$\term(\ctxs{\Gamma}{\Deltab{0}}{\Thetab{0}})$}, iff \textnormal{$\df(\ctxs{\Gamma}{\Deltab{0}}{\Thetab{0}})$} and $\exists k$ finite s.t. $\forall n \geq k : \ctxs{\Gamma}{\Deltab{0}}{\Thetab{0}} \redP \cdots \redP \ctxs{\Gamma}{\Deltab{n}}{\Thetab{n}}$ implies \textnormal{$\pEnd(\Deltab{n})$}.
    \end{enumerate}
\end{definition}

The type-level equivalents of the process properties are given in \cref{def:type-props} above.
Concretely, a context is considered \emph{deadlock free} iff it only stops reducing because its linear part is \tEnd-typed.
A context is \emph{terminating} iff it is deadlock free \emph{and} it reaches the state in which it can no longer reduce within a finite number of reduction steps.

\begin{theorem}[Property Verification: \textnormal{\df}\ and \textnormal{\term}]\label{thm:prop-ver-full}
    If $\ctxs{\Gamma}{\Delta}{\Theta} \vdash \net$ with $\safe_\R(\ctxs{\Gamma}{\Delta}{\Theta})$, then $\phi(\ctxs{\Gamma}{\Delta}{\Theta})$ implies $\phi(\net)$, for $\phi \in \{\textnormal{\df}, \textsf{\term}\}$.
\end{theorem}

\begin{proof}
    \Cref{thm:prop-ver-full} follows directly from \cref{thm:sf}. 
    We know that if a context can reduce, then the network can match at least one reduction.
    Hence to prove property verification, we infer the shape of a network typed under the definition of $\phi$ on types, and show that this fits the definition of $\phi$ described for networks.
    We elaborate the case of \df\ below, the other case is similar.

    From $\df(\ctxs{\Gamma}{\Delta}{\Theta})$ it follows that there are two subcases:
    \begin{enumerate}
        \item Consider $\ctxs{\Gamma}{\Delta}{\Theta}$ can reduce, \ie\ $\ctxs{\Gamma}{\Delta}{\Theta} \redP$.
        Then from $\ctxs{\Gamma}{\Delta}{\Theta} \vdash \net$ and \cref{thm:sf} we infer that $\net$ can reduce, which satisfies $\df(\net)$ trivially.

        \item Consider $\ctxs{\Gamma}{\Delta}{\Theta}$ cannot reduce, \ie\ $\ctxs{\Gamma}{\Delta}{\Theta} \not\redP$.
        Then, from the definition of \df\ we know $\pEnd(\Delta)$. 
        If the linear context is \tEnd-typed then we can infer the shape of the entire network to be \emph{at most} some composition of replicated receives with the network buffer.
        This shape of $\net$ fits the definition of $\df$ in \cref{def:proc-props}, and thus we can conclude $\df(\net)$.
    \end{enumerate}\qed
\end{proof}

\end{document}